\documentclass[conference]{IEEEtran}
\IEEEoverridecommandlockouts
 \usepackage{lipsum,graphicx,multicol}
\usepackage{color,soul}
\usepackage{mathtools}
\usepackage{graphicx}
\usepackage{color,soul}
\usepackage{setspace}
\usepackage{multirow}
\usepackage{tablefootnote}
\usepackage[norelsize, linesnumbered, ruled, lined, boxed, commentsnumbered]{algorithm2e}
\usepackage[normalem]{ulem}
\usepackage[font=small,labelfont=bf]{caption}
\usepackage{bm, amsmath, amssymb, amsthm}
\usepackage{amsfonts}
\usepackage {amssymb,amsmath,xcolor}
\usepackage{booktabs}
\usepackage{cite}
\usepackage{array}
\usepackage[acronym]{glossaries}
\usepackage{cases}
\usepackage{cleveref}

\def\BibTeX{{\rm B\kern-.05em{\sc i\kern-.025em b}\kern-.08em
    T\kern-.1667em\lower.7ex\hbox{E}\kern-.125emX}}
    
  \newtheorem{theorem}{Theorem}

\newtheorem{proposition}{Proposition}
\newtheorem{lemma}{Lemma}
\newtheorem{rem}{Remark}

\newcommand{\bieee}{\begin{IEEEeqnarray}{rCl}}
\newcommand{\eieee}{\end{IEEEeqnarray}}

\newcommand{\nn}{\nonumber}

\SetKwInput{KwInput}{Input}                
\SetKwInput{KwOutput}{Output}              

\begin{document}

\title{Non-Coherent Fast-Forward Relays for Full-Duplex Jamming Attack
}

\author{\IEEEauthorblockN{Vivek Chaudhary, Harshan Jagadeesh}
\IEEEauthorblockA{Department of Electrical Engineering, Indian Institute of Technology, Delhi, India\\
chaudhary03vivek@gmail.com, jharshan@gmail.com}
}

\maketitle

\begin{abstract}
This work addresses a strategy to mitigate jamming attack on low-latency communication by a Full-Duplex (FD) adversary in fast-fading channel conditions. The threat model is such that the FD adversary can jam a frequency band and measure the jammed band's power level. We first point out that due to the presence of this FD adversary, Frequency Hopping (FH) fails. We then propose a fast-forward cooperative relaying scheme, wherein the victim node hops to the frequency band of a nearby FD helper node that fast-forwards the victim's symbol along with its symbol. At the same time, the victim and the helper cooperatively pour some fraction of their power on the jammed band to engage the adversary. Due to fast-fading channel conditions, the victim and the helper use amplitude based non-coherent signalling referred to as Non-Coherent Fast-Forward Full-Duplex (NC-F2FD) relaying. To minimize the error-rate of this strategy, we jointly design the constellations at the helper node and the victim node by formulating an optimization problem. Using non-trivial results, we first analyse the objective function and then, based on the analytical results, propose a low-complexity algorithm to synthesize the fast-forwarded constellations. Through simulations, we show that the error performance of the victim improves after employing our countermeasure.
\end{abstract}


\section{Introduction}

Critical delay requirements in the industrial automation system, medical industry and automotive applications in 5G technology are such that $<$1ms end-to-end delay is expected across layers \cite{IoT}. Due to stringent latency constraints, it becomes imperative to secure such applications from an attack that introduces a delay in message delivery. One such simple, yet effective attack is the Denial of Service (DoS) attack~\cite{DoS}, where the adversary reduces the Signal-to-Interferece-Noise Ratio (SINR) at the receiver below the required threshold by injecting excess noise. Although Frequency Hopping (FH) is known to combat DoS attack, it requires the victim to hop out of the operating frequency, thus dropping the power level on the jammed band. Moreover, with advancements in Self-Interference Cancellation (SIC) techniques in Full-Duplex (FD) radios \cite{FD1}, a jammer can use an FD radio for jamming~\cite{Hanwal} and to detect a drop in power level~\cite{my_PIMRC}, \cite{my_TCCN}. For instance, \cite{my_PIMRC} and \cite{my_TCCN} model an adversary which uses in-band FD radio to jam a frequency band while simultaneously measuring the power fluctuation on the jammed band. The jammer, in this case, detects a countermeasure if it experiences a drop in the power level of the jammed frequency band. \cite{my_PIMRC} and \cite{my_TCCN} also discuss the mitigation techniques to combat FD jamming attacks using fast-forward cooperative communication. Here, the victim node and the FD helper node use non-coherent and coherent modulation techniques, respectively.

The countermeasure provided by \cite{my_PIMRC} and \cite{my_TCCN} assume a slow-fading channel; thus, the helper uses coherent modulation techniques. However, when channel conditions vary rapidly, coherent detection is challenging. Therefore, the limitations of FH in mitigating FD jamming attacks followed by the challenges in performing coherent detection in fast-fading channel motivates us to synthesize countermeasures that combat FD jamming adversary in fast-fading channel conditions. Therefore, in this paper, we propose a Non-Coherent Fast-Forward Full-Duplex (NC-F2FD) relaying scheme, wherein both the victim and the helper use non-coherent modulation schemes.

\subsection{Contributions}
The proposed NC-F2FD relaying protocol is such that the victim moves out of the jammed band and shares the uplink frequency with a nearby FD helper node. Further, the victim and the helper use modified constellations (from their original constellation) for communicating with the base-station. Meanwhile, the base-station performs joint decoding of both users. The focus of this paper is to choose the constellation points at the victim node as well as the helper such that, the error performance of the victim improves with increasing Signal-to-Noise Ratio (SNR). The contributions of this paper are as follows:
\begin{itemize}
\item For the proposed NC-F2FD scheme, we formulate a constrained optimization problem to minimize the joint probability of error at the base-station over the modified constellations of the victim and the helper node.
\item Using non-trivial analysis, we study the behaviour of the objective function in terms of modified constellations, and based on the behaviour of objective function, we propose a low-complexity algorithm to obtain the modified constellations.
\item We present the joint error performance of two nodes as a function of SNR when the constellation points are obtained using the proposed algorithm. We also present the joint error performance when the constellation points are obtained using exhaustive search and show that the error performance of the two methods overlap. Finally, when using NC-F2FD, we show that the error performance of the victim improves with increasing SNR.
\end{itemize} 

\subsection{Related work and Novelty}
 Our work can be viewed as a constellation design problem for a non-coherent fast-forward strategy. With respect to related work on constellation design, \cite{Ranjan} and \cite{Goldsmith2} compute the optimal constellation points for multilevel non-coherent signalling in Rayleigh fading channel for point-to-point communication. Recently, \cite{Joint-MAC} and \cite{Goldsmith1} addressed the problem of constellation design for non-coherent Multiple Access Channels (MAC). However, due to fast-forward FD relaying, our work cannot be viewed as a direct extension of~\cite{Ranjan}, \cite{Goldsmith2}, \cite{Joint-MAC}, or \cite{Goldsmith1}. The work closest to our idea are \cite{my_PIMRC} and \cite{my_TCCN}, where the authors assume an FD \emph{jam and measure} adversary. However, \cite{my_PIMRC} and \cite{my_TCCN} assume a slow-fading channel between the nodes and base-station, whereas this work considers a fast-fading channel. Thus, the overall system model is entirely different.

\section{System Model}

\begin{figure}
\centering
\includegraphics[width = 8cm, height = 5cm]{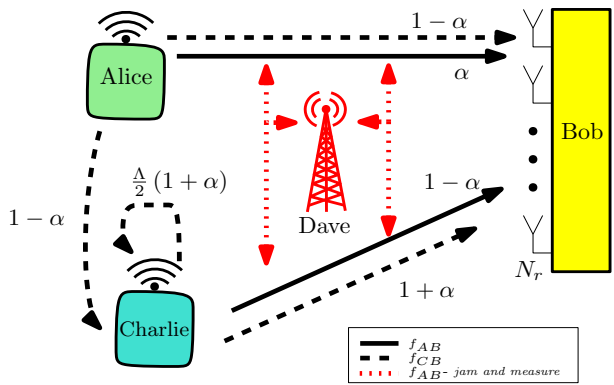}
\caption{\label{fig:sys_model} System model comprising Alice (the victim), Charlie (the helper), Bob (base-station) and Dave (FD adversary). Alice and Charlie cooperatively pour $0.5(1-\alpha)$ and $0.5(1+\alpha)$ on $f_{CB}$, respectively. The residual energy in poured on $f_{AB}$, which is constantly monitored by Dave. $\Lambda$ represents the SIC achieved at Charlie.}
\vspace{-0.5cm}
\end{figure}

We consider a network model with multiple single-antenna transmitting nodes and a base-station with $N_{r}$ antennas. The transmitting nodes have heterogeneous demands. Further, the symbols from these nodes experience fast-fading channels. \figurename{\ref{fig:sys_model}} captures one such instance where two nodes, Alice and Charlie convey their messages on orthogonal frequencies, $f_{AB}$ and $f_{CB}$, respectively, to the multi-antenna base-station, Bob. Alice has messages with low-latency constraint, while Charlie has messages with no such low-latency constraints. Since low-latency messages are accompanied by low data-rate (e.g., control channel messages (PUCCH) in 5G), we assume that Alice uses two-level non-coherent Amplitude Shift Keying (ASK). Additionally, due to fast-fading channel conditions, Charlie also resorts to a non-coherent ASK signalling scheme. Although Alice and Charlie use non-coherent ASK symbols, the average power on these nodes may differ due to heterogeneous demands. We now introduce an FD \emph{jam and measure} adversary, Dave, who injects high powered noise on $f_{AB}$ to introduce delay in Alice's low-latency messages. Since Dave is equipped with an FD radio, he can simultaneously measure the power level on $f_{AB}$ besides jamming $f_{AB}$. Thus, in the presence of such a sophisticated adversary, if Alice chooses to use a completely different band for transmission, Dave will measure a considerable dip in the power level of $f_{AB}$, thus detecting the countermeasure. The detection of countermeasure at Dave will compel him to attack any of the remaining frequencies, further degrading the network performance. Therefore, a countermeasure must ensure reliable reception of Alice's symbols at Bob using a different frequency band, while maintaining the power level on the jammed frequency band same as before implementing the countermeasure. In the next section, we present a novel relaying strategy wherein Alice, with the help of Charlie, communicate her low-latency messages to Bob by sharing Charlie's uplink frequency through a non-coherent strategy. 

Throughout the paper, $(\cdot)^{H}$ represents Hermitian of a vector; bold face represent a vector. Further, $\Gamma(a,x)$, $\gamma(a,x)$, and $\Gamma(x)$  are upper incomplete Gamma function, lower incomplete Gamma function, and gamma function, respectively.

\section{Non-Coherent Fast-Forward Full-Duplex Relaying (NC-F2FD)}

In the proposed relaying scheme, Bob directs Alice to communicate her messages on $f_{CB}$, i.e., the uplink frequency of a nearby node, Charlie. Charlie is a single-antenna transceiver supporting in-band FD communication on $f_{CB}$, and being a legitimate node in the network, has symbols to convey to Bob. Therefore, Charlie listens and decodes Alice's symbol instantaneously on $f_{CB}$, multiplexes the decoded symbol into his symbol, and transmits the multiplexed symbol on $f_{CB}$ to Bob. With this multiplexing strategy, Bob observes multiple access channel on $f_{CB}$ and attempts to decode Alice's and Charlie's symbols jointly. Notice that if, Alice utilises all her power for communicating on $f_{CB}$, Dave measures a significant dip in the power level on $f_{AB}$ and therefore this countermeasure is not effective. As a result, Alice and Charlie must ensure that the power level on $f_{AB}$ remains the same as before applying the countermeasure. To accomplish this, Alice and Charlie cooperatively pour only a fraction of their power on $f_{CB}$, using $\alpha\in(0,1)$, and the residual power from the nodes are injected on $f_{AB}$. Henceforth, we refer to $\alpha$ as the power splitting factor. Furthermore, since both links (Alice-to-Charlie, Charlie-to-Bob) are non-coherent, we refer to our relaying scheme as Non-Coherent Fast-Forward Full-Duplex (NC-F2FD) scheme. The aim of NC-F2FD is to help Alice reliably communicate her messages to Bob on $f_{CB}$ with acceptable degradation in the error performance of Charlie. We will focus on the error analysis on $f_{CB}$ at Charlie and Bob in the upcoming sections.

\subsection{Signal Model:NC-F2FD}
Before using the cooperative relaying scheme, let Alice and Charlie choose symbols from a scaled version of $\mathcal{S} = \{0,1\}$. Here, the scale factor determines the average transmit power, such that the average powers are $0.5$ and $1$ for Alice and Charlie, respectively. When Alice and Charlie start cooperating, Alice transmits $x\in\left\{0, \sqrt{1-\alpha}\right\}$ on $f_{CB}$ for bit $0$ and bit $1$, respectively. Therefore, the baseband symbol received at Charlie is,

\bieee
r_{C} = h_{AC}x + h_{CC}y' + n_{C}, \label{eq:charlie_rx}
\eieee

\noindent where, $h_{AC}$ and $n_{C}$ are the baseband channel for Alice-to-Charlie link and the additive white Gaussian noise (AWGN) at Charlie, distributed as ${\cal CN}(0,\sigma_{AC}^{2})$ and ${\cal CN}(0,N_{o})$, respectively. Further, the Self-Interference (SI) channel between Charlie's transmit and receive antenna is defined by $h_{CC}\sim{\cal CN}\left(0,\Lambda\frac{(1+\alpha)}{2}\right)$\cite{my_TCCN}, such that $\Lambda$ is the degree of SIC achieved by Charlie's FD radio. Finally, $y'$ is the multiplexed symbol transmitted by Charlie chosen from the Charlie's modified constellation, $\mathcal{S}_{C}$. Here, $y'$ is a function of the original symbol from Charlie, i.e., $y\in\{0,1\}$ and the decoded symbol from Alice, $\hat{x}$. After decoding $x$, Charlie uses the following rule for multiplexing; firstly, the ordering of the elements of $\mathcal{S}_{C}$ should follow  \emph{Gray Mapping}. Furthermore, to ensure that he makes minimum compromise while helping Alice, after decoding $\hat{x}=0$, Charlie transmits symbols with energy on the extreme ends. Hence, Charlie uses the following criterion for multiplexing $\hat{x}$ to his transmitting symbols,

\begin{subnumcases}{y'=}
\sqrt{\epsilon_{1}}, & $\hat{x} = 0$, $y = 0$,\label{eq:Rule00}
\\
\sqrt{\alpha\eta_{1}}, & $\hat{x} = 1$, $y = 0$,\label{eq:Rule10}
\\
\sqrt{\alpha\eta_{2}}, & $\hat{x} = 1$, $y = 1$,\label{eq:Rule11}
\\
\sqrt{\epsilon_{2}}, & $\hat{x} = 0$, $y = 1$.\label{eq:Rule01}
\end{subnumcases}

The mapping also includes the power splitting factor $\alpha$ in the transmit symbol of Charlie. Finally, it is evident that $\sqrt{\epsilon_{1}}$ and $\sqrt{\epsilon_{2}}$ should have minimum and maximum energies, respectively, are per the criterion. Moreover, the energy transmitted by Charlie when $y=0$ should be less than when $y=1$, therefore, $\sqrt{\alpha\eta_{1}}<\sqrt{\alpha\eta_{2}}$. Further, due to Gray encoding at Charlie, $\sqrt{\epsilon_{1}}$ and $\sqrt{\alpha\eta_{2}}$ cannot be adjacent neighbours. Therefore, we have  $\sqrt{\epsilon_{1}}<\sqrt{\alpha\eta_{1}}< \sqrt{\alpha\eta_{2}}<\sqrt{\epsilon_{2}}$, where, $\epsilon_{1}, \epsilon_{2}, \eta_{1}, \eta_{2} >0$.
 
Initially, the average power on $f_{CB}$ was $1$, and after NC-F2FD, Alice pours an average power of $0.5\left(1-\alpha\right)$ on $f_{CB}$. Thus, the rest of $0.5\left(1+\alpha\right)$ must be poured by Charlie. Similarly on $f_{AB}$, Alice and Charlie pour $0.5\alpha$ and $0.5\left(1-\alpha\right)$, respectively, to manage an average power equal to $0.5$. Further, based on the signal model of NC-F2FD, if Charlie does not make any error in decoding Alice's symbol, the average power on $f_{AB}$ is such that Dave does not observe any power fluctuations on $f_{AB}$. Now, to achieve the average power constraint on $f_{CB}$, the elements of $\mathcal{S}_{C}$ must follow the equality constraint:

\bieee
\dfrac{1}{4}\left(\epsilon_{1} + \alpha\eta_{1} + \alpha\eta_{2} + \epsilon_{2}\right) &=& \dfrac{1}{2}(1+\alpha).\label{eq:constraint}
\eieee

\noindent Finally, the baseband symbol received at Bob is given by,

\bieee\label{eq:rB}
\mathbf{r}_{B} = \mathbf{h}_{AB}x + \mathbf{h}_{CB}y' + \mathbf{n}_{B},
\eieee
where, $y' = \left\{\sqrt{\epsilon_{1}}, \sqrt{\alpha\eta_{1}}, \sqrt{\alpha\eta_{2}}, \sqrt{\epsilon_{2}}\right\}$ is chosen by Charlie based on the decoded symbol $\hat{x}$ and his original symbol $y$ as given by~\eqref{eq:Rule00}--\eqref{eq:Rule01}. We assume $N_{r}$ receive antennas at Bob, therefore, Alice-to-Bob link and Charlie-to-Bob link are $\mathbf{h}_{AB}\sim{\cal} (\mathbf{0}_{N_{r}},\sigma_{AB}^{2}\mathbf{I}_{N_{r}})$ and $\mathbf{h}_{CB}\sim{\cal} (\mathbf{0}_{N_{r}},\sigma_{CB}^{2}\mathbf{I}_{N_{r}})$, respectively, such that $\sigma_{AB}^{2}=\sigma_{CB}^{2}=1$.  Further, the AWGN at Bob is given by, $\mathbf{n}_{B}\sim{\cal} (\mathbf{0}_{N_{r}},N_{o}\mathbf{I}_{N_{r}})$. We assume all the channel realizations and noise samples are statistically independent. Further, various noise variances at Charlie and Bob are given by $N_{o} = (\text{SNR})^{-1}$. In the next section, we analyse the error performance when decoding Alice's symbols at Charlie as it is important for evaluating the joint error performance at Bob.

\subsection{Error analysis at Charlie}

Charlie computes the threshold, $\tau$ for energy detection based on the received symbol $r_{C}$ in~\eqref{eq:charlie_rx}. The threshold is given by $\tau = \frac{N_{C0}N_{C1}}{N_{C0}-N_{C1}}\ln\left(\frac{N_{C0}}{N_{C1}}\right)$, where, $N_{C0} = N_{o} + 0.5\Lambda(1+\alpha)$ and $N_{C1} = \sigma_{AC}^{2}(1-\alpha) + 0.5\Lambda(1+\alpha)+ N_{o}$ are the energy levels corresponding to the two energy levels at Alice. Due to vicinity of Alice and Charlie, we assume that $\sigma_{AC}^{2}>\sigma_{CB}^{2}$. Now, the decision rule for decoding is, $r_{C}r_{C}^{H} \overset{0}{\underset{1}{\lessgtr}} \tau$. Therefore, the probability of bit $0$ decoded as bit $1$ and vice-versa, is given by $\mathsf{P}_{01} = e^{-\frac{\tau}{N_{C0}}}$ and $\mathsf{P}_{10} = 1-e^{-\frac{\tau}{N_{C1}}}$, respectively. As a result, probability that bit $0$ is correctly decoded at Charlie is, $\mathsf{P_{00}} = 1-\mathsf{P_{01}}$, and the probability of that of bit $1$ is, $\mathsf{P_{11}} = 1 - \mathsf{P_{10}}$. 

Next, we recollect very important results that determine the behaviour of $\mathsf{P_{11}}$, $\mathsf{P_{00}}$, $\mathsf{P_{01}}$, and $\mathsf{P_{10}}$ as a function of $\alpha$.

\begin{lemma}\label{lm:0110}
The error probabilities at Charlie are such that $\mathsf{P_{10}}<\mathsf{P_{11}}$, $\forall\alpha\in(0,\nu)$, where $\nu = \frac{\sigma_{AC}^{2}-N_{o}-\frac{\Lambda}{2}}{\sigma_{AC}^{2}+\frac{\Lambda}{2}}$, and $\mathsf{P_{00}}>\mathsf{P_{01}}$ $\forall\alpha\in(0,1)$.\cite{my_TCCN}
\end{lemma}

\begin{lemma}\label{lm:p10p01}
$\mathsf{P_{10}}>\mathsf{P_{01}}$,  $\forall\alpha\in(0,1)$.\cite{my_TCCN}
\end{lemma}


\begin{lemma}\label{lm:p11}
For $\alpha\!\in\!(0,1)$, $\mathsf{P_{11}}$ and $\mathsf{P_{00}}$ are decreasing function of $\alpha$, hence, $\mathsf{P_{10}}$ and $\mathsf{P_{01}}$ are increasing function of $\alpha$.\cite{my_TCCN}
\end{lemma}

\begin{rem}\label{rem:p00p11}
For a given SNR and $N_{r}$, the approximation $\ln\frac{\mathsf{P_{11}}}{\mathsf{P_{00}}}\approx 0$ is tight when $\alpha$ is away from $1$.
\end{rem}

\subsection{Error Analysis at Bob}

\begin{figure}
\centering
\includegraphics[scale=0.2]{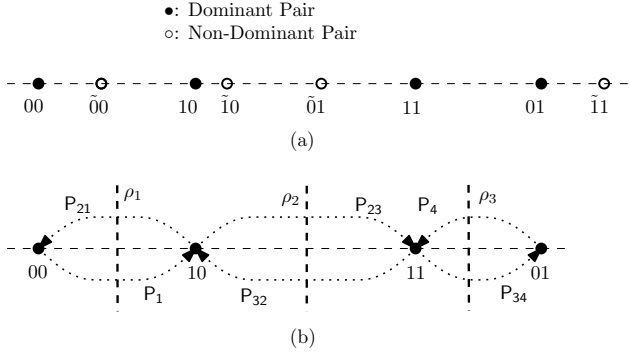}
\caption{\label{fig:constellation}(a) Constellation symbols received at Bob. (b) Simplified constellation as result of dominant error terms.}
\end{figure}

Bob uses joint Maximum A Posteriori (MAP) decoder to perform non-coherent energy detection and jointly decode Alice's and Charlie's symbols from $\mathbf{r}_{B}$ in~\eqref{eq:rB}. The joint MAP decoder at Bob is defined as,

\bieee\label{eq:MAP}
\hat{i},\hat{j} = \underset{i,j\in \{0,1\}}{\arg \max}\ f\left(\left.\mathbf{r}_{B}\right\vert x=i,y=j\right),
\eieee

\begin{figure*}[!btp]
\begin{small}
\bieee
f\!\left(\left.\!\mathbf{r}_{B}\right\vert x\!=\!0,y\!=\!0\right)\! &=&\! \mathsf{P}_{00}g\!\left(\left.\!\mathbf{r}_{B}\right\vert x\!=\!0,y'\!=\!\sqrt{\!\epsilon_{1}}\right)\! +\! \mathsf{P}_{01}g\left(\left.\mathbf{r}_{B}\right\vert x=0,y'=\sqrt{\!\alpha\eta_{1}}\right)\!\approx\!\mathsf{P}_{00}g\left(\left.\mathbf{r}_{B}\right\vert x\!=\!0,y'\!=\!\sqrt{\!\epsilon_{1}}\right)\!\triangleq\! f_{D}\!\left(\left.\!\mathbf{r}_{B}\right\vert x\!=\!0,y\! =\!0\right)\IEEEeqnarraynumspace \label{eq:f00}\\
f\!\left(\left.\!\mathbf{r}_{B}\right\vert x\!=\!0,y\!=\!1\right)\! &=&\! \mathsf{P}_{00}g\!\left(\left.\!\mathbf{r}_{B}\right\vert x\!=\!0,y'\!=\!\sqrt{\!\epsilon_{2}}\right)\! +\! \mathsf{P}_{01}g\left(\left.\mathbf{r}_{B}\right\vert x=0,y'=\sqrt{\!\alpha\eta_{2}}\right)\!\approx\!\mathsf{P}_{00}g\left(\left.\mathbf{r}_{B}\right\vert x\!=\!0,y'\!=\!\sqrt{\!\epsilon_{2}}\right)\!\triangleq\! f_{D}\!\left(\left.\!\mathbf{r}_{B}\right\vert x\!=\!0,y\! =\!1\right)\IEEEeqnarraynumspace \label{eq:f01}\\
f\!\left(\left.\!\mathbf{r}_{B}\right\vert x\!=\!1,y\!=\!0\right)\! &=&\! \mathsf{P}_{11}g\!\left(\left.\!\mathbf{r}_{B}\right\vert x\!=\!1,y'\!=\!\sqrt{\!\alpha\eta_{1}}\right)\! +\! \mathsf{P}_{10}g\left(\left.\mathbf{r}_{B}\right\vert x=1,y'=\sqrt{\!\epsilon_{1}}\right)\!\approx\!\mathsf{P}_{11}g\left(\left.\mathbf{r}_{B}\right\vert x\!=\!1,y'\!=\!\sqrt{\!\alpha\eta_{1}}\right)\!\triangleq\! f_{D}\!\left(\left.\!\mathbf{r}_{B}\right\vert x\!=\!1,y\! =\!0\right)\IEEEeqnarraynumspace \label{eq:f10}\\
f\!\left(\left.\!\mathbf{r}_{B}\right\vert x\!=\!1,y\!=\!1\right)\! &=&\! \mathsf{P}_{11}g\!\left(\left.\!\mathbf{r}_{B}\right\vert x\!=\!1,y'\!=\!\sqrt{\!\alpha\eta_{2}}\right)\! +\! \mathsf{P}_{10}g\left(\left.\mathbf{r}_{B}\right\vert x=1,y'=\sqrt{\!\epsilon_{2}}\right)\!\approx\!\mathsf{P}_{11}g\left(\left.\mathbf{r}_{B}\right\vert x\!=\!1,y'\!=\!\sqrt{\!\alpha\eta_{2}}\right)\!\triangleq\! f_{D}\!\left(\left.\!\mathbf{r}_{B}\right\vert x\!=\!1,y\! =\!1\right)\IEEEeqnarraynumspace \label{eq:f11}
\eieee
\end{small}
\hrulefill
\end{figure*}

\noindent where, $f\left(\left.\mathbf{r}_{B}\right\vert x=i,y=j\right)$ is the probability density function (pdf) of $\mathbf{r}_{B}$ conditioned on the realizations of $x$ and $y$. The pdf described in \eqref{eq:MAP} is a Gaussian mixture as shown in~\eqref{eq:f00}--\eqref{eq:f11}, where, $g(\cdot)$ denotes the pdf of $\mathbf{r}_{B}$ conditioned on $x$, $\hat{x}$, and $y$. Additionally, the first equality of~\eqref{eq:f00}--\eqref{eq:f11} provide realizations of $f(\cdot)$, for various combinations of $i$ and $j$. Therefore, to obtain the average probability of error at Bob using the joint MAP decoder, one needs to solve pairwise error probability that include Gaussian mixtures. However, solving  Gaussian mixtures for computing the average probability of error is non-trivial. Therefore, we propose the Joint Dominant Decoder (JDD), an approximate joint MAP decoder that considers only the dominant terms in the Gaussian mixtures of~\eqref{eq:f00}--\eqref{eq:f11}. Using Lemma~\ref{lm:0110}, JDD is defined as,

\bieee\label{eq:JD}
\hat{i},\hat{j} = \underset{i,j\in \{0,1\}}{\arg \max}\ f_{D}\left(\left.\mathbf{r}_{B}\right\vert x=i,y=j\right),
\eieee

\noindent where, we dropped the coefficients of $\mathsf{P}_{01}$ and $\mathsf{P}_{10}$ in~\eqref{eq:f00}--\eqref{eq:f11} to obtain the approximation $f_{D}(\cdot)$. In the next section, we compute the probability of various error events when using JDD.

\section{NC-F2FD Joint Dominant Decoder} 

In this section, we use JDD defined in the previous section to compute the  average probability of error at Bob. A joint error event is defined as an error at Bob in decoding a transmitted pair $(i\ j)$ as either $(\overline{i}\ j)$, $(i\ \overline{j})$ or $(\overline{i}\ \overline{j})$, where, $i$ and $j$ denote the symbols transmitted by Alice and Charlie, respectively, such that $i,j\in\left\{0,1\right\}$. Further, $\overline{i}$ and $\overline{j}$ denote the complement of $i$ and $j$, respectively. Therefore, we have 4 pairs corresponding to $(i,j)\in\{0,1\}\times\{0,1\}$, and Bob receives symbol corresponding to one of these pairs. However, due to decode-and-forward relaying scheme, a transmitted symbol $x$ from Alice can be decoded with error at Charlie as well as at Bob. We first consider error events at Charlie, and assume that Alice transmits $x=1$ for exposition. Subsequently, we assume that Charlie decodes it incorrectly as $0$, and then he chooses either $y' = \sqrt{\epsilon_{1}}$ or $\sqrt{\epsilon_{2}}$, instead of $y' = \sqrt{\alpha\eta_{1}}$ or $y' = \sqrt{\alpha\eta_{2}}$. Similarly, when Alice transmits $x=0$, and if Charlie decodes it as $1$, he transmits $y' = \sqrt{\alpha\eta_{1}}$ or $y' = \sqrt{\alpha\eta_{2}}$ rather than $y' = \sqrt{\epsilon_{1}}$ or $\sqrt{\epsilon_{2}}$. Meanwhile, when $x=1$ or $x=0$, Alice contributes $1-\alpha$ or \emph{zero} power, respectively, to the received symbol at Bob. Therefore, we have a total of $8$ symbols received at Bob, $4$ each for Charlie's correct and incorrect decisions. In~\figurename{\ref{fig:constellation}~(a)}, we present the energy levels received at Bob corresponding to $\left\{(00),(10),(11),(01)\right\}$ as well as $\left\{(\overline{0}0), (\overline{0}1), (\overline{1}0), (\overline{1}1)\right\}$. Formally, based on~\eqref{eq:Rule00}--\eqref{eq:Rule01} and \eqref{eq:rB}, the received symbol at Bob is a function of $x$, $\hat{x}$ and $y$, hence distributed as,

\begin{equation}\label{eq:rB1}
\mathbf{r}_{B}\!\sim\!
\begin{cases}
 {\cal CN}\left(\mathbf{0}_{N_{r}}, \upsilon_{00}\mathbf{I}_{N_{r}}\right),   & x=0,\hat{x}=0,y = 0;
\\
{\cal CN}\left(\mathbf{0}_{N_{r}}, \upsilon_{01}\mathbf{I}_{N_{r}}\right), & x=0,\hat{x}=0,y=1;
\\
 {\cal CN}\left(\mathbf{0}_{N_{r}}, \upsilon_{\overline{0}0}\mathbf{I}_{N_{r}}\right), &  x=0,\hat{x}=1,y=0;
\\
{\cal CN}\left(\mathbf{0}_{N_{r}}, \upsilon_{\overline{0}1}\mathbf{I}_{N_{r}}\right), &  x=0,\hat{x}=1,y=1;
\\
{\cal CN}\left(\mathbf{0}_{N_{r}}, \upsilon_{\overline{1}0}\mathbf{I}_{N_{r}}\right),  &  x=1,\hat{x}=0,y=0;
\\
{\cal CN}\left(\mathbf{0}_{N_{r}}, \upsilon_{\overline{1}1}\mathbf{I}_{N_{r}}\right), & x=1,\hat{x}=0,y=1;
\\
 {\cal CN}\left(\mathbf{0}_{N_{r}}, \upsilon_{10}\mathbf{I}_{N_{r}}\right), &  x=1,\hat{x}=1,y=0;
\\
{\cal CN}\left(\mathbf{0}_{N_{r}}, \upsilon_{11}\mathbf{I}_{N_{r}}\right), &  x=1,\hat{x}=1,y=1.
\end{cases}
\end{equation}

\noindent where, $\upsilon_{00} = \epsilon_{1} + N_{o}$, $\upsilon_{10} = 1-\alpha + \alpha\eta_{1}+N_{o}$, $\upsilon_{11} = 1-\alpha + \alpha\eta_{2}+N_{o}$, $\upsilon_{01} = \epsilon_{2} + N_{o}$, $\upsilon_{\overline{1}0} = 1-\alpha + \epsilon_{1} + N_{o}$, $\upsilon_{\overline{1}1} = 1-\alpha + \epsilon_{2} + N_{o}$, $\upsilon_{\overline{0}1} =  \alpha\eta_{2} + N_{o}$, and $\upsilon_{\overline{0}0} =  \alpha\eta_{1} + N_{o}$, are the variances of the corresponding received symbols at Bob, representing the received pair.

It is well known that the decision statistics for non-coherent energy detection, is given by, $\mathbf{r}_{B}\mathbf{r}_{B}^{H}\sim \text{Gamma}(N_{r},\upsilon)$ where, $\text{Gamma}(\cdot,\cdot)$ is the Gamma distribution\cite{Ranjan} and $\upsilon$ is the variance of the received symbol. Therefore, based on the encoding rule at Charlie, the received energy $\mathbf{r}_{B}\mathbf{r}_{B}^{H}$, follow the order as shown in \figurename{\ref{fig:constellation}~(a)}. We now use JDD that only considers the dominant pairs $\left\{(00),(10),(11),(01)\right\}$ for error computation by neglecting the non-dominant pairs arising due to Charlie's erroneous decisions, i.e., $\left\{(\overline{0}0), (\overline{0}1), (\overline{1}0), (\overline{1}1)\right\}$. The simplified constellation as received by JDD is shown in \figurename{\ref{fig:constellation}~(b)}.
 
To analyse the error events, let a pair $(i\ j)$ be transmitted and decoded as $(i^{*}\ j^{*})$, where, $(i^{*}\ j^{*})\in\{0,1\}\times\{0,1\}$. Using JDD defined in~\eqref{eq:JD}, a pair is in error when $(i\ j)\neq(i^{*}\ j^{*})$, and is defined as, $\Delta_{(i\ j)\rightarrow(i^{*}\ j^{*})}$.

\bieee
\Delta_{(i\ j)\rightarrow(i^{*}\ j^{*})} = \dfrac{f_{D}\left(\left.\mathbf{r}_{B}\right\vert x= i, y=j\right)}{f_{D}\left(\left.\mathbf{r}_{B}\right\vert x= i^{*}, y=j^{*}\right)} \leq 1,\nn
\eieee

\noindent where, $f_{D}\left(\cdot\right)$ is defined in \eqref{eq:f00}-\eqref{eq:f11} for each transmitted pair $(i\ j)$.
Although, JDD considers only the dominant pairs for error computation, the received symbols at Bob could be any one of the non-dominant pairs too. Therefore, the event $\Delta_{(i\ j)\rightarrow(i^{*}\ j^{*})}$ is conditioned on whether the decision statistics, $\mathbf{r}_{B}\mathbf{r}_{B}^{H}$ at Bob corresponds to a dominant pair or a non-dominant pair, i.e., $``\left.\Delta_{(i\ j)\rightarrow(i^{*}\ j^{*})}\right\vert\upsilon_{(i\ j)}"$ or $``\left.\Delta_{(i\ j)\rightarrow(i^{*}\ j^{*})}\right\vert\upsilon_{(\overline{i}\ j)}"$, respectively. Here, $\left.\Delta_{(i\ j)\rightarrow(i^{*}\ j^{*})}\right\vert\upsilon_{(\overline{i}\ j)}$ is known as the counterpart of $\left.\Delta_{(i\ j)\rightarrow(i^{*}\ j^{*})}\right\vert\upsilon_{(i\ j)}$ that arises due to error at Charlie. Thus, we formally define an error event as,

\bieee
\Pr((i\ j)\rightarrow(i^{*}\ j^{*})) &=& \mathsf{P}_{i\ i}\Pr\left(\left.\Delta_{(i\ j)\rightarrow(i^{*}\ j^{*})}\right\vert \upsilon_{i\ j}\right)+ \nn \\
&\ & \mathsf{P}_{i\ \overline{i}}\Pr\left(\left.\Delta_{(i\ j)\rightarrow(i^{*}\ j^{*})}\right\vert \upsilon_{\overline{i}\ j}\right).\label{eq:error_event}
\eieee

Further, since the decision statistics $\mathbf{r}_{B}\mathbf{r}_{B}^{H}\in\mathbb{R}^{+}$, for a given pair $(i\ j)$, we consider its adjacent neighbours for error computation. Therefore, from \figurename{\ref{fig:constellation}~(b)}, all the possible error events are, $(00)\rightarrow(10)$, $(10)\rightarrow(00)$, $(10)\rightarrow(11)$, $(11)\rightarrow(10)$, $(11)\rightarrow(01)$ and $(01)\rightarrow(11)$, and the probability of each error event is,  $\mathsf{P_{1}}:\Pr\left((00)\rightarrow(10)\right)$; $\mathsf{P_{21}}:\Pr\left((10)\rightarrow(00)\right)$; $\mathsf{P_{23}}:\Pr\left((10)\rightarrow(11)\right)$; $\mathsf{P_{32}}:\Pr\left((11)\rightarrow(10)\right)$; $\mathsf{P_{34}}:\Pr\left((11)\rightarrow(01)\right)$; and $\mathsf{P_{4}}:\Pr\left((01)\rightarrow(11)\right)$. Note that we also have a counterpart associated with each error event whose probabilities are given by $\mathsf{P_{1}^{C}}$, $\mathsf{P_{21}^{C}}$, $\mathsf{P_{34}^{C}}$, and $\mathsf{P_{4}^{C}}$, similar to \eqref{eq:error_event}. 

Next, to compute the probability of error between dominant pairs, we need to compute the detection thresholds as shown in~\figurename{\ref{fig:constellation}~(b)}. Let $\rho_{1}$, $\rho_{2}$, and $\rho_{3}$ be the 3 thresholds between the 4 dominant pairs. Using first principles, the expression for $\rho_{1}$, $\rho_{2}$, and $\rho_{3}$ can be computed as,

\bieee
\rho_{1} &=& \frac{\upsilon_{00}\upsilon_{10}}{\upsilon_{00}-\upsilon_{10}}\left[N_{r}\ln\left(\frac{\upsilon_{00}}{\upsilon_{10}}\right) + \ln\left(\frac{\mathsf{P}_{11}}{\mathsf{P}_{00}}\right)\right]\approx\ \ \ \ \ \nn\\
&&\hspace{1.3in} N_{r}\frac{\upsilon_{00}\upsilon_{10}}{\upsilon_{00}-\upsilon_{10}}\ln\left(\frac{\upsilon_{00}}{\upsilon_{10}}\right),\label{eq:rho1}\\
\rho_{2} &=& N_{r}\frac{\upsilon_{10}\upsilon_{11}}{\upsilon_{10}-\upsilon_{11}}\ln\left(\frac{\upsilon_{10}}{\upsilon_{11}}\right),\label{eq:rho2}\\
\rho_{3} &=& \frac{\upsilon_{11}\upsilon_{01}}{\upsilon_{11}-\upsilon_{01}}\left[N_{r}\ln\left(\frac{\upsilon_{11}}{\upsilon_{01}}\right) + \ln\left(\frac{\mathsf{P}_{00}}{\mathsf{P}_{11}}\right)\right] \approx\ \ \ \ \ \nn\\
&&\hspace{1.25in} N_{r}\frac{\upsilon_{11}\upsilon_{01}}{\upsilon_{11}-\upsilon_{01}}\ln\left(\frac{\upsilon_{11}}{\upsilon_{01}}\right).\label{eq:rho3}
\eieee

In practice, Bob can receive both dominant as well as non-dominant pairs. However, JDD assumes that Bob receives only dominant pairs and computes the detection thresholds using the dominant terms. Therefore, for computing the error when Bob receives non-dominant pairs, i.e., $\mathsf{P_{1}^{C}}$, $\mathsf{P_{21}^{C}}$, $\mathsf{P_{34}^{C}}$, and $\mathsf{P_{4}^{C}}$, we use the same detection thresholds, i.e., $\rho_{1}$, $\rho_{2}$, and $\rho_{3}$. In the following proposition, we present probability of error event $(00)\rightarrow(10)$, i.e., $\Pr((00)\rightarrow(10))$.
 
\begin{proposition}\label{pr:p1}

\bieee
\Pr((00)\rightarrow(10)) &=& \mathsf{P}_{00}\Pr\left(\left.\Delta_{(00)\rightarrow(10)}\right\vert \upsilon_{00}\right)+ \nn \\
&\ & \mathsf{P}_{01}\Pr\left(\left.\Delta_{(00)\rightarrow(10)}\right\vert \upsilon_{\overline{0}0}\right)\nn\\
&=& \mathsf{P}_{00}\mathsf{P}_{1} + \mathsf{P}_{01}\mathsf{P_{1}^{C}}\nn,
\eieee

\noindent where, $\mathsf{P}_{1} = \frac{\Gamma\left(N_{r}, \frac{\rho_{1}}{\upsilon_{00}}\right)}{\Gamma\left(N_{r}\right)}$ and $\mathsf{P_{1}^{C}} = \frac{\Gamma\left(N_{r}, \frac{\rho_{1}}{\upsilon_{\overline{0}0}}\right)}{\Gamma\left(N_{r}\right)}$.
\end{proposition}

On similar lines, we compute all the error events and tabulate them in~\tablename{~\ref{tab:Pe_tab}}.

\begin{table}[h]
\begin{center}
\begin{tabular}{ | m{2.3cm}| m{5cm} | } 
\hline
Error Event  & \multicolumn{1}{c|}{Probability}  \\ 
\hline
$\Pr\left((00)\rightarrow(10)\right)$ & $\mathsf{P}_{00}\mathsf{P}_{1} + \mathsf{P}_{01}\mathsf{P_{1}^{C}}$, where $\mathsf{P}_{1} = \frac{\Gamma\left(N_{r}, \frac{\rho_{1}}{\upsilon_{00}}\right)}{\Gamma\left(N_{r}\right)}$ and $\mathsf{P_{1}^{C}} = \frac{\Gamma\left(N_{r}, \frac{\rho_{1}}{\upsilon_{\overline{0}0}}\right)}{\Gamma\left(N_{r}\right)}$ \\ 
\hline
$\Pr\left((10)\rightarrow(00)\right)$ & $\mathsf{P}_{11}\mathsf{P}_{21} + \mathsf{P}_{10}\mathsf{P_{21}^{C}}$, where $\mathsf{P}_{21}=\frac{\gamma\left(N_{r}, \frac{\rho_{1}}{\upsilon_{10}}\right)}{\Gamma\left(N_{r}\right)}$ and $\mathsf{P_{21}^{C}} = \frac{\gamma\left(N_{r}, \frac{\rho_{1}}{\upsilon_{\overline{1}0}}\right)}{\Gamma\left(N_{r}\right)}$ \\ 
\hline
$\Pr\left((10)\rightarrow(11)\right)$ & $\mathsf{P}_{11}\mathsf{P}_{23}$, where $\mathsf{P}_{23}=\frac{\Gamma\left(N_{r}, \frac{\rho_{2}}{\upsilon_{10}}\right)}{\Gamma\left(N_{r}\right)}$ \\ 
\hline
$\Pr\left((11)\rightarrow(10)\right)$ & $\mathsf{P}_{11}\mathsf{P}_{32}$, where $\mathsf{P}_{32}=\frac{\gamma\left(N_{r}, \frac{\rho_{2}}{\upsilon_{11}}\right)}{\Gamma\left(N_{r}\right)}$ \\ 
\hline
$\Pr\left((11)\rightarrow(01)\right)$ & $\mathsf{P}_{11}\mathsf{P}_{34} + \mathsf{P}_{10}\mathsf{P_{34}^{C}}$, where $\mathsf{P}_{34}=\frac{\Gamma\left(N_{r}, \frac{\rho_{3}}{\upsilon_{11}}\right)}{\Gamma\left(N_{r}\right)}$ and $\mathsf{P_{34}^{C}} = \frac{\gamma\left(N_{r}, \frac{\rho_{3}}{\upsilon_{\overline{1}1}}\right)}{\Gamma\left(N_{r}\right)}$ \\ 
\hline
$\Pr\left((01)\rightarrow(11)\right)$ & $\mathsf{P}_{00}\mathsf{P}_{4} + \mathsf{P}_{01}\mathsf{P_{4}^{C}}$, where $\mathsf{P}_{4} = \frac{\Gamma\left(N_{r}, \frac{\rho_{3}}{\upsilon_{01}}\right)}{\Gamma\left(N_{r}\right)}$ and $\mathsf{P_{4}^{C}} = \frac{\Gamma\left(N_{r}, \frac{\rho_{3}}{\upsilon_{\overline{0}1}}\right)}{\Gamma\left(N_{r}\right)}$  \\ 
\hline
\end{tabular}
\end{center}
\caption{\label{tab:Pe_tab}Error events with their respective probabilities.}
\end{table}

Considering equally likely information symbols at Alice and Charlie, the average probability of error at Bob is given by $\mathsf{P_{e}}$,

\bieee
\mathsf{P_{e}} \!\!=\!\!\dfrac{1}{|\mathcal{S}_{C}|}\sum_{i=0}^{1}\sum_{j=0}^{1}\Pr((i\ j)\rightarrow(\overline{i}\ j))\!+\!\Pr((i\ j)\rightarrow(\overline{i}\ \overline{j}))\label{eq:pe},\IEEEeqnarraynumspace
\eieee

\noindent where, the events $(00)\rightarrow(11)$ and $(11)\rightarrow(00)$ are invalid. Now substituting the probability computed for each error event from~\tablename{~\ref{tab:Pe_tab}}, and upper-bounding the complementary error event, i.e., $\mathsf{P_{1}^{C}}$, $\mathsf{P_{21}^{C}}$, $\mathsf{P_{34}^{C}}$, and $\mathsf{P_{4}^{C}}$ by $1$, an upper bound on the average probability is given by,

\begin{multline}\label{eq:pe1}
\mathsf{P}_{e}\!\leq\!\mathsf{P_{e}^{\star}}\!=\! \dfrac{1}{4}\left[\mathsf{P}_{00}\left(\mathsf{P_{1}}+\mathsf{P_{4}}\right) + 2\mathsf{P_{01}} + 2\mathsf{P_{10}} + \right. \\
\ \left. \mathsf{P_{11}}\left(\mathsf{P_{21}} + \mathsf{P_{23}}+\mathsf{P_{32}} + \mathsf{P_{34}}\right)\right].
\end{multline}

In the next section, we find the values of $\left\{\epsilon_{1}, \epsilon_{2}, \eta_{1}, \eta_{2}, \alpha\right\}$ that minimizes $\mathsf{P_{e}^{\star}}$. Interestingly, we can visualise $\sqrt{\epsilon_{1}}$, $\sqrt{\epsilon_{2}}$, $\sqrt{\alpha\eta_{1}}$, and $\sqrt{\alpha\eta_{2}}$ as the modified constellation symbols at Charlie, such that each symbol contains information about Alice's as well as Charlie's symbol. Thus, we can identify this work as a constellation design problem for the NC-F2FD relaying scheme.

\section{Optimization of Amplitude Levels}

We now formulate the optimization problem to synthesize the amplitude levels $\sqrt{\epsilon_{1}}$, $\sqrt{\epsilon_{2}}$, $\sqrt{\alpha\eta_{1}}$, and $\sqrt{\alpha\eta_{2}}$, subject to the average power constraint at Charlie. Our formulation minimises the joint probability of error at Bob as given by,

\begin{equation}\label{eq:opt}
\begin{aligned}
\min_{\epsilon_{1}, \epsilon_{2}, \eta_{1}, \eta_{2}, \alpha} \quad & \mathsf{P_{e}^{\star}}\\
\textrm{s.t.} \quad & \dfrac{1}{4}\left(\epsilon_{1} + \alpha\eta_{1} + \alpha\eta_{2} + \epsilon_{2}\right) = \dfrac{1}{2}(1+\alpha),\\
  & 0<\alpha<1,   \\
\end{aligned}
\end{equation}

\noindent where, $\mathsf{P_{e}^{\star}}$ is the upper-bound on the average probability of error obtained using JDD. To solve the optimization problem in~\eqref{eq:opt}, one can solve a multivariate Lagrangian equation. However, it is well known that solving a multivariate Lagrangian equation is hard. Therefore, in the next section, we take a different approach towards obtaining $\left\{\epsilon_{1}, \epsilon_{2}, \eta_{1}, \eta_{2}, \alpha\right\}$ by analysing $\mathsf{P_{e}^{\star}}$ using some non-trivial relations.

%
%

\subsection{Variation of $\mathsf{P_{e}^{\star}}$}

In this section, we analytically study the behaviour of the objective function by varying only a subset of the variables amongst $\left\{\epsilon_{1}, \epsilon_{2}, \eta_{1}, \eta_{2}, \alpha\right\}$, while keeping the rest of them fixed. We aim to reduce the dimensionality of the search space $\left\{\epsilon_{1}, \epsilon_{2}, \eta_{1}, \eta_{2}, \alpha\right\}$, thereby reducing the implementation complexity.

\begin{rem}
From Proposition~\ref{pr:p1}, $\mathsf{P_{e}^{\star}}$ is minimum when $\upsilon_{00}$ is minimum, i.e., $\epsilon_{1}=0$, for any $\eta_{1}, \eta_{2}, \epsilon_{2}, \alpha$ satisfying the constraint in~\eqref{eq:opt}.
\end{rem}

\begin{rem}\label{rem:eta1}
Note that when $\eta_{1}$ increases beyond $1$, the energy difference between pairs $(11)$ and $(01)$, and pairs $(11)$ and $(10)$ reduces, thereby increasing $\mathsf{P_{e}^{\star}}$, therefore, $0\leq\eta_{1}<1$.
\end{rem}

We rearrange the average power constraint in~\eqref{eq:constraint} as, $\epsilon_{2} = 2 - \alpha\left(\eta_{1} + \eta_{2} - 2\right)$ and substitute it in~\eqref{eq:pe1}. With, $\epsilon_{1}=0$ and $\epsilon_{2}$ expressed in terms of $\alpha, \eta_{1}$, and $\eta_{2}$, the expression of $\mathsf{P_{e}^{\star}}$ is now only a function of $\alpha$, $\eta_{1}$, and $\eta_{2}$ for a fixed $N_{o} $ and $N_{r}$. Thus, we have reduced the dimension of search space to 3. In general, we cannot comment on the nature of $\mathsf{P_{e}^{\star}}$ as a function of $\alpha$, $\eta_{1}$, and $\eta_{2}$. Therefore, we analyse $\mathsf{P_{e}^{\star}}$ in single dimension by fixing $\eta_{1}$, $\alpha$ and then observe the nature of $\mathsf{P_{e}^{\star}}$ w.r.t. $\eta_{2}\in \mathbb{R^{+}}$. Similarly, for a fixed $\eta_{1}$ and $\eta_{2}$, we observe the nature of $\mathsf{P_{e}^{\star}}$ w.r.t. $\alpha\in (0,1)$. Along these lines, we determine the nature of all the error events in~\tablename{~\ref{tab:Pe_tab}} as a function of $\alpha$ and $\eta_{2}$ by keeping $\eta_{1}$ fixed.

\begin{lemma}
For a fixed $\eta_{1}$ and $\alpha$, $\mathsf{P_{21}}$ is independent of $\eta_{2}$. Also, $\mathsf{P_{21}}$ is an increasing function of $\alpha$. However, $\mathsf{P_{23}}$ decreases when either $\eta_{2}$ increases and $\alpha$ is fixed or when $\alpha$ increases and $\eta_{2}$ is kept constant.
\end{lemma}

\begin{proof}
In the first part, we will prove that $\mathsf{P_{21}}$ is independent of $\eta_{2}$. The term $\mathsf{P_{21}}$ is given by, 

\bieee
\mathsf{P_{21}} &=& \dfrac{\gamma\left(N_{r}, \frac{\rho_{1}}{\upsilon_{10}}\right)}{\Gamma\left(N_{r}\right)}.\nn 
\eieee

\bieee
\noalign{\noindent\text{From~\eqref{eq:rho1}}, we have,}\dfrac{\rho_{1}}{\upsilon_{10}} &=& N_{r}\dfrac{\upsilon_{00}}{\upsilon_{00}-\upsilon_{10}}\ln\left(\dfrac{\upsilon_{00}}{\upsilon_{10}}\right),\nn\\
&=& N_{r}\dfrac{\upsilon_{00}}{\upsilon_{10}-\upsilon_{00}}\ln\left(\dfrac{\upsilon_{10}}{\upsilon_{00}}\right),\nn\\
&=& N_{r}\dfrac{\ln\left(1+\kappa_{1}\right)}{\kappa_{1}}\label{eq:r1_v10},
\eieee

\noindent where, $\kappa_{1} = \frac{\upsilon_{10}-\upsilon_{00}}{\upsilon_{00}} = \frac{1-\alpha+\alpha\eta_{1}}{N_{o}}$. Since $\upsilon_{10}>\upsilon_{00}$, we have $\kappa_{1}>1$. We observe that, $\kappa_{1}$ is independent of $\eta_{2}$, therefore, $\mathsf{P_{21}}$ is independent of $\eta_{2}$.

Furthermore, differentiating $\kappa_{1}$ w.r.t $\alpha$ gives us $(-1+\eta_{1})N_{o}^{-1}$. The term $(-1+\eta_{1})N_{o}^{-1}<0$, since $\eta_{1}<1$ (Remark~\ref{rem:eta1}); thus, $\kappa_{1}$ is a decreasing function of $\alpha$. In addition, $\frac{\ln\left(1+\kappa_{1}\right)}{\kappa_{1}}$ is a decreasing function of $\kappa_{1}>0$. Therefore, $N_{r}\frac{\ln\left(1+\kappa_{1}\right)}{\kappa_{1}}$ and hence the ratio $\frac{\rho_{1}}{\upsilon_{10}}$ in~\eqref{eq:r1_v10} is an increasing function of $\alpha$. Furthermore, since $\gamma\left(N_{r}, \frac{\rho_{1}}{\upsilon_{10}}\right)$ is an increasing function of $\frac{\rho_{1}}{\upsilon_{10}}$, as $\alpha$ increases, $\gamma\left(N_{r}, \frac{\rho_{1}}{\upsilon_{10}}\right)$ increases. Thus,  $\mathsf{P_{21}}$ is an increasing function of $\alpha$. 

\noindent On similar lines, $\mathsf{P_{23}}$ is given by,

\bieee
\mathsf{P_{23}} &=& \dfrac{\Gamma\left(N_{r}, \frac{\rho_{2}}{\upsilon_{10}}\right)}{\Gamma\left(N_{r}\right)}\nn
\eieee

\noindent where, 

\bieee
\dfrac{\rho_{2}}{\upsilon_{10}} &=& N_{r}\frac{\upsilon_{11}}{\upsilon_{10}-\upsilon_{11}}\ln\left(\frac{\upsilon_{10}}{\upsilon_{11}}\right),\nn\\
&=& N_{r}\dfrac{\ln\left(1+\kappa_{2}\right)}{\kappa_{2}}\label{eq:r2_v10},
\eieee

\noindent such that, $\kappa_{2} = \frac{\upsilon_{10}-\upsilon_{11}}{\upsilon_{11}} = \frac{\alpha\left(\eta_{1}-\eta_{2}\right)}{1-\alpha(1-\eta_{2})}$. Since, $\upsilon_{10}<\upsilon_{11}$, $\kappa_{2}\in(-1,0)$. Taking derivative of $\kappa_{2}$ w.r.t. $\eta_{2}$ gives $\frac{-\alpha\upsilon_{10}}{\upsilon_{11}^{2}}$. This shows that $\kappa_{2}$ is a decreasing function of $\eta_{2}$. Now, as $\eta_{2}$ increases, $\kappa_{2}$ decreases and therefore, the ratio in~\eqref{eq:r2_v10}, $\frac{\rho_{2}}{\upsilon_{10}} = N_{r}\frac{\ln\left(1+\kappa_{2}\right)}{\kappa_{2}}$ increases. Since, $\Gamma\left(N_{r}, \frac{\rho_{2}}{\upsilon_{10}}\right)$ is a decreasing function of $\frac{\rho_{2}}{\upsilon_{10}}$, $\Gamma\left(N_{r}, \frac{\rho_{2}}{\upsilon_{10}}\right)$ decreases as $\frac{\rho_{2}}{\upsilon_{10}}$ increases. Thus, $\mathsf{P_{23}}$ is a decreasing function of $\eta_{2}$. 

Similarly, differentiating $\kappa_{2}$ w.r.t. $\alpha$ gives $\frac{\eta_{1}-\eta_{2}}{\left(1-\alpha + \alpha\eta_{2}\right)^{2}}<0$, implying $\kappa_{2}$ is a decreasing function of $\alpha$. Now, as $\alpha$ increases, $\kappa_{2}$ decreases and therefore, the ratio in~\eqref{eq:r2_v10}, $\frac{\rho_{2}}{\upsilon_{10}} = N_{r}\frac{\ln\left(1+\kappa_{2}\right)}{\kappa_{2}}$ increases. Since, $\Gamma\left(N_{r}, \frac{\rho_{2}}{\upsilon_{10}}\right)$ is a decreasing function of $\frac{\rho_{2}}{\upsilon_{10}}$, $\Gamma\left(N_{r}, \frac{\rho_{2}}{\upsilon_{10}}\right)$ decreases as $\frac{\rho_{2}}{\upsilon_{10}}$ increases. Thus, $\mathsf{P_{23}}$ is a decreasing function of $\alpha$. 
\end{proof}

Along similar lines, we can prove the behaviour of each term in~\tablename{~\ref{tab:Pe_tab}}, as either increasing, decreasing or independent functions of $\alpha$ and $\eta_{2}$ when the other is fixed. The same is summarized in~\tablename{~\ref{tab:Tab}}.

\begin{table}[h]
\begin{center}
\begin{tabular}{ c|c|c|c|c|c|c| }
\cline{2-7}
 & $\mathsf{P_{1}}$ & $\mathsf{P_{21}}$ & $\mathsf{P_{23}}$ & $\mathsf{P_{32}}$ & $\mathsf{P_{34}}$ & $\mathsf{P_{4}}$ \\
\hline 
\multicolumn{1}{|l|}{$\eta_{2}$} & \multirow{1}{*}{$\bigstar$} & \multirow{1}{*}{$\bigstar$} & $-$ & $-$ & $+$ & $+$\\  
\hline
\multicolumn{1}{|l|}{$\alpha$} & $+$ & $+$ & $-$ & $-$ & $-$ & $-$ \\
\hline
\end{tabular}
\end{center}
\caption{\label{tab:Tab}Table depicting the behaviour of error events. For each error event, the symbol ``+'', ``-'', and ``$\bigstar$'' represent, increasing, decreasing, and independent behaviours, as a function of $\left\{\alpha,\eta_{1},\eta_{2}\right\}$, such that the entry in the first column is variable while the rest two are fixed.}
\end{table}

\begin{figure}[h]
\begin{center}
\includegraphics[width = 0.5\textwidth, height = 0.4\textheight]{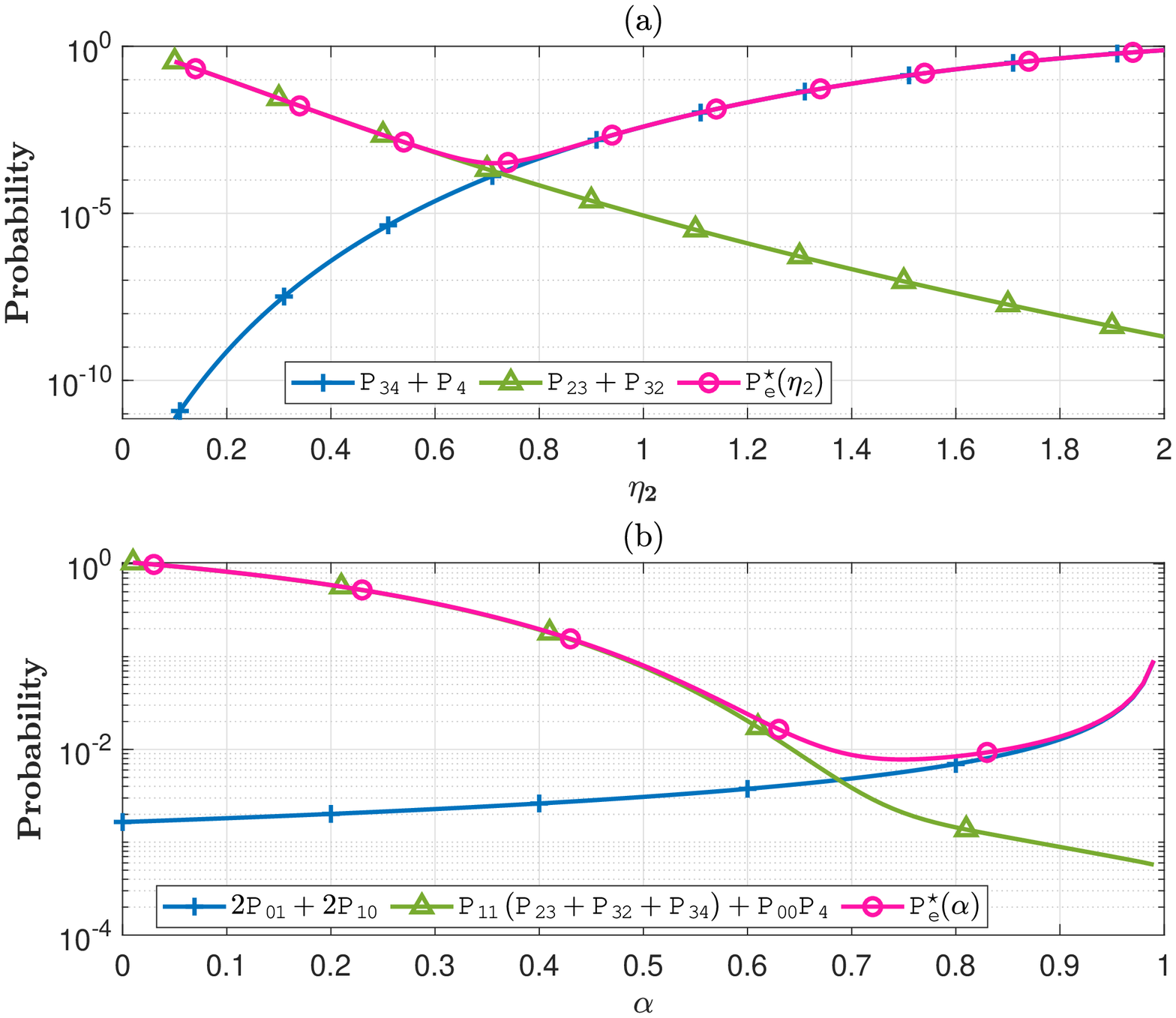}
\vspace{-0.8cm}
\caption{\label{fig:unimod}For SNR = 35dB, $N_{r} = 32$, (a) The dip experienced by $\mathsf{P_{e}^{\star}}$ as a function of $\eta_{2}$ and the intersection of $\mathsf{P_{34}} + \mathsf{P_{4}}$ and $\mathsf{P_{23}} + \mathsf{P_{32}}$ are approximately close for fixed $\eta_{1}$ and $\alpha$. (b) The dip experienced by $\mathsf{P_{e}^{\star}}$ as a function of $\alpha$ and the intersection of $\mathsf{P_{11}}\left(\mathsf{P_{34}} + \mathsf{P_{32}} + \mathsf{P_{23}}\right) + \mathsf{P_{00}}\mathsf{P_{4}}$ and $2\mathsf{P_{01}} + 2\mathsf{P_{10}}$ are approximately close for fixed $\eta_{1}$ and $\alpha$.} 
\end{center}
\end{figure}

From~\tablename{~\ref{tab:Tab}}, the components of $\mathsf{P_{e}^{*}}$ are either monotonically decreasing or monotonically increasing as a function of $\eta_{2}$ or $\alpha$. In~\figurename{\ref{fig:unimod}}~(a), we plot $\mathsf{P_{e}^{*}}$, and the increasing and decreasing components of $\mathsf{P_{e}^{*}}$ as a function of $\eta_{2}$. Similarly, in~\figurename{\ref{fig:unimod}}~(b), we plot $\mathsf{P_{e}^{*}}$, and the increasing and decreasing components of $\mathsf{P_{e}^{*}}$ as a function of $\alpha$. We observe that the increasing and decreasing components of $\mathsf{P_{e}^{*}}$ in both cases intersect only once. Moreover, the intersection of the components in both cases is very close to the value of $\eta_{2}$ ($\alpha$) at which $\mathsf{P_{e}^{*}}$ experiences minima. Thus, in the next two theorems, for a constant $\eta_{1}$, we prove that the increasing and decreasing components of $\mathsf{P_{e}^{\star}}$ intersect only once when either $\eta_{2}$ or $\alpha$ is varied.

\begin{theorem}\label{th:uni_eta2}
For a given $\eta_{1}$ and $\alpha$, the decreasing and increasing components of $\mathsf{P_{e}^{\star}}$ in~\eqref{eq:pe1}, i.e., $\mathsf{P_{23}} + \mathsf{P_{32}}$ and  $\mathsf{P_{34}} + \mathsf{P_{4}}$, respectively, intersect only once for $\eta_{2}\in\left(\eta_{1},0.5\left(3+\alpha^{-1}-\eta_{1}\right)\right)$.
\end{theorem}

\begin{proof}
When $\alpha$ is fixed, the terms $\mathsf{P_{1}}$, $\mathsf{P_{21}}$, $\mathsf{P_{00}}$, $\mathsf{P_{11}}$, $\mathsf{P_{01}}$ and $\mathsf{P_{10}}$ in~\eqref{eq:pe1} are independent of $\eta_{2}$. From Table~\ref{tab:Tab}, $\mathsf{P_{23}}$ and $\mathsf{P_{32}}$ are decreasing functions of $\eta_{2}$, and $\mathsf{P_{34}}$ and $\mathsf{P_{4}}$ are increasing functions of $\eta_{2}$.

We have $\eta_{2}\in\left(\eta_{1},0.5\left(3+\alpha^{-1}-\eta_{1}\right)\right)$ since, $\upsilon_{10}<\upsilon_{11}<\upsilon_{01}$. For  $\mathsf{P_{23}}+\mathsf{P_{32}}$ and  $\mathsf{P_{4}}+\mathsf{P_{34}}$ to have one intersection, we can straightaway prove that $\mathsf{P_{23}} + \mathsf{P_{32}}$ and $\mathsf{P_{34}} + \mathsf{P_{4}}$ have their order reversed at the extreme values of $\eta_{2}$. We first evaluate $\mathsf{P_{23}} + \mathsf{P_{32}}$ at extreme lower side of $\eta_{2}$ to get,

\begin{small}
\bieee
\lim_{\eta_{2}\!\rightarrow\eta_{1}}\!\!\mathsf{P_{23}}\! +\! \mathsf{P_{32}}\! &=&\!\! \dfrac{\Gamma\left(N_{r},N_{r}\right)}{N_{r}}\!\! +\!\! \dfrac{\gamma\left(N_{r},N_{r}\right)}{N_{r}}\!=\! 1.\nn
\eieee

\noindent However, the value of $\mathsf{P_{23}}+\mathsf{P_{32}}$ at the right most extreme of $\eta_{2}$ is given by,

\bieee
\lim_{\eta_{2}\rightarrow 0.5\left(3\!+\alpha^{-1}\!-\!\eta_{1}\right)} \mathsf{P_{23}} + \mathsf{P_{32}}\!\! &=&\!\! \dfrac{\Gamma\left(N_{r},\frac{\rho_{2}}{\upsilon_{10}}\right)}{N_{r}}\!\! +\!\! \dfrac{\gamma\left(N_{r},\frac{\rho_{2}}{\upsilon_{11}}\right)}{N_{r}}\!\!<1.\nn
\eieee
\end{small}

\noindent Similarly,

\begin{small}
\bieee
\lim_{\eta_{2}\rightarrow\eta_{1}}\!\!\mathsf{P_{34}}\! +\! \mathsf{P_{4}}\! &=&\!\! \dfrac{\Gamma\left(N_{r},\frac{\rho_{3}}{\upsilon_{11}}\right)}{N_{r}}\!\! +\!\! \dfrac{\gamma\left(N_{r},\frac{\rho_{3}}{\upsilon_{01}}\right)}{N_{r}}\!\!<1,\nn\\
\lim_{\eta_{2}\rightarrow 0.5\left(3\!+\!\alpha^{-1}-\!\eta_{1}\right)}\!\!\mathsf{P_{34}}\! +\! \mathsf{P_{4}}\! &=&\! \dfrac{\Gamma\left(N_{r},N_{r}\right)}{N_{r}}\! +\! \dfrac{\gamma\left(N_{r},N_{r}\right)}{N_{r}}\! =\! 1.\nn
\eieee
\end{small}

\noindent Altogether, we have,

\begin{small}
\begin{equation*}
\begin{cases}
 \mathsf{P_{23}} + \mathsf{P_{32}} > \mathsf{P_{34}} + \mathsf{P_{4}},  & \eta_{2}\rightarrow\eta_{1},
\\
\mathsf{P_{23}} + \mathsf{P_{32}} < \mathsf{P_{34}} + \mathsf{P_{4}}, & \eta_{2}\rightarrow 0.5\left(3+\alpha^{-1}-\eta_{1}\right).
\end{cases}
\end{equation*}
\end{small}

\noindent Therefore, the increasing and decreasing components in $\mathsf{P_{e}^{\star}}$ intersect only once as a function of $\eta_{2}$ and this intersection can be computed using the Newton-Raphson (NR) \cite{NR} algorithm. This completes the proof.
\end{proof}

\begin{theorem}\label{th:uni_alpha}
For a given $\eta_{1}$ and $\eta_{2}$, the decreasing and increasing components of $\mathsf{P_{e}^{\star}}$ in~\eqref{eq:pe1}, i.e., $\mathsf{P}_{00}\left(\mathsf{P_{1}}+\mathsf{P_{4}}\right) + \mathsf{P_{11}}\left(\mathsf{P_{21}} + \mathsf{P_{23}}+\mathsf{P_{32}} + \mathsf{P_{34}}\right)$ and $2\mathsf{P_{01}} + 2\mathsf{P_{10}}$, respectively, intersect only once for $\alpha\in(0,1)$.
\end{theorem}

\begin{proof}
Rearranging~\eqref{eq:pe1}, we get 

\begin{multline*}
\mathsf{P_{e}^{\star}} = \dfrac{1}{4}\left[\mathsf{P_{00}}\mathsf{P_{1}} + \mathsf{P_{11}}\mathsf{P_{21}} + \mathsf{P_{11}}\left(\mathsf{P_{23}} + \mathsf{P_{32}} + \mathsf{P_{34}}\right) +\right.\\
\ \left.\mathsf{P_{00}}\mathsf{P_{4}} + 2\mathsf{P_{01}} + 2\mathsf{P_{10}}\right].
\end{multline*}

\noindent We make the following observations
\begin{itemize}
\item Since each term in $\mathsf{P_{11}}\left(\mathsf{P_{23}} + \mathsf{P_{32}} + \mathsf{P_{34}}\right) + \mathsf{P_{00}}\mathsf{P_{4}}$ is a decreasing function of $\alpha$, the overall sum is a decreasing function of $\alpha$. Further, as $\alpha\rightarrow 1$, $\mathsf{P_{11}}\left(\mathsf{P_{23}} + \mathsf{P_{32}} + \mathsf{P_{34}}\right) + \mathsf{P_{00}}\mathsf{P_{4}}\approx 0$, since, $\mathsf{P_{00}} = \mathsf{P_{11}}\approx 0$ as $\alpha\rightarrow 1$. 
\item For $\epsilon_{1} = 0$, when $\alpha\rightarrow 0$, $\mathsf{P_{00}}\mathsf{P_{1}} + \mathsf{P_{11}}\mathsf{P_{21}} \approx 0$, since $\mathsf{P_{1}} = \mathsf{P_{21}} \approx 0$ when, $\alpha\rightarrow 0$. Similarly, when $\alpha\rightarrow 1$,  $\mathsf{P_{00}}\mathsf{P_{1}} + \mathsf{P_{11}}\mathsf{P_{21}} \approx 0$, since $\mathsf{P_{00}} = \mathsf{P_{11}}\approx 0$, when $\alpha\rightarrow 1$.
\item From Lemma~\ref{lm:p11}, $2\mathsf{P_{01}} + 2\mathsf{P_{10}}$ monotonically increases with $\alpha$, and as $\alpha\rightarrow 1$, $\mathsf{P_{01}}=\mathsf{P_{10}}\approx 1$.
\end{itemize}

\noindent Thus, it is elementary to show that,

\bieee
\mathsf{P_{11}}\left(\mathsf{P_{23}} + \mathsf{P_{32}} + \mathsf{P_{34}}\right) + \mathsf{P_{00}}\mathsf{P_{4}} &>& 2\mathsf{P_{01}} + 2\mathsf{P_{10}}, \alpha\rightarrow 0,\nn\\
\mathsf{P_{11}}\left(\mathsf{P_{23}} + \mathsf{P_{32}} + \mathsf{P_{34}}\right) + \mathsf{P_{00}}\mathsf{P_{4}} &<& 2\mathsf{P_{01}} + 2\mathsf{P_{10}}, \alpha\rightarrow 1.\nn
\eieee 

\noindent Further, the intersection can be computed using NR \cite{NR} algorithm. This completes the proof.
\end{proof}

In the next section, we use Theorem~\ref{th:uni_eta2} and Theorem~\ref{th:uni_alpha} to present a low-complexity algorithm. Using this algorithm, we obtain the local minima of $\eta_{2}$ and $\alpha$ for a given $\eta_{1}$.

\subsection{Algorithm}

\begin{algorithm}
  
  \KwInput{$\mathsf{P_{e}}$ from~\eqref{eq:pe1}, $\delta_{\mathsf{P_{e}^{\star}}}>0$, $\delta_{\eta_{1}}>0$}
  \KwOutput{$\left\{\eta_{1}^{\dagger}, \eta_{2}^{\dagger},\alpha^{\dagger}\right\}$}
  Initialize: $\eta_{1}\gets 0$, $\eta_{2}\gets \eta_{2}^{o}$, $\alpha\gets \alpha^{o}$\\
  \While{true} 
  {
   $\mathsf{P_{e}^{o}} \gets \mathsf{P_{e}^{\star}}\left(\alpha,\eta_{1},\eta_{2}\right)$\\
  \While{true}
  {
  Find $\eta_{2}^{i}$ using Theorem~\ref{th:uni_eta2} and update $\mathsf{P_{e}^{\eta_{2}}} \gets \mathsf{P_{e}^{\star}}\left(\eta_{1}, \eta_{2}^{i},\alpha\right)$\\
 Find $\alpha^{i}$ using Theorem~\ref{th:uni_alpha} and update $\mathsf{P_{e}^{\alpha}} \gets \mathsf{P_{e}^{\star}}\left(\eta_{1}, \eta_{2},\alpha^{i}\right)$\\
  \If{$\mathsf{P_{e}^{\alpha}}-\mathsf{P_{e}^{\eta_{2}}} \geq \delta_{\mathsf{P_{e}}}$}
  {
  $\eta_{2} \gets \eta_{2}^{i}$, $\alpha \gets \alpha^{o}$\\
  continue
  }
  \ElseIf{$\mathsf{P_{e}^{\alpha}}-\mathsf{P_{e}^{\eta_{2}}} \leq -\delta_{\mathsf{P_{e}}}$}
  {
  $\eta_{2} \gets \eta_{2}^{o}$, $\alpha \gets \alpha^{i}$\\
  continue
  }
  \ElseIf {$\left\vert\mathsf{P_{e}^{\alpha}}-\mathsf{P_{e}^{\eta_{2}}}\right\vert<\delta_{\mathsf{P_{e}}}$}
  {
  $\mathsf{P_{e}^{\iota}}= \min\left(\mathsf{P_{e}^{\alpha}}, \mathsf{P_{e}^{\eta_{2}}}\right)$\\
  break
  }
  }
  \If{$\left\vert\mathsf{P_{e}^{\iota}}-\mathsf{P_{e}^{o}}\right\vert > \delta_{\mathsf{P_{e}}}$}
  {
  $\eta_{1} \gets \eta_{1} + \delta_{\eta_{1}}$, $\mathsf{P_{e}^{o}}\gets \mathsf{P_{e}^{\iota}}$\\
  continue
  }
  \Else
  {
  $\eta_{1}^{\dagger} \gets \eta_{1}$, $\eta_{2}^{\dagger} \gets \eta_{2}^{i}$, $\alpha^{\dagger} \gets \alpha^{i}$\\
  break
  }
  }
  \caption{\label{Algo} Two-Layer Greedy Descent  Algorithm}
\end{algorithm}

In the previous section, we proved that for a fixed $\eta_{1}$ and $\alpha$, the value of $\eta_{2}$ at which the increasing and decreasing components in $\mathsf{P_{e}^{\star}}$  intersect is very close to the value of $\eta_{2}$ at which $\mathsf{P_{e}^{\star}}$ is minimum. Similar argument is also valid when $\eta_{1}$ and $\eta_{2}$ are fixed while, $\alpha$ is the variable. We exploit these properties to compute $\left\{\eta_{1}, \eta_{2}, \alpha\right\}$ such that $\mathsf{P_{e}^{\star}}$ is evaluated at local minima for a given SNR and $N_{r}$.

In the proposed algorithm, as presented in Algorithm~\ref{Algo}, we fix $\epsilon_{1} = 0$ and initialise $\eta_{1} = 0$. We also initialise $\eta_{2}$ and $\alpha$ with arbitrary values $\eta_{2}^{o}$ and $\alpha^{o}$, respectively. Using the initial values, the algorithm computes $\mathsf{P_{e}^{o}}$ using \eqref{eq:pe1}. The algorithm then obtains $\eta_{2}^{i}$ and $\alpha^{i}$ using Theorem~\ref{th:uni_eta2} and Theorem~\ref{th:uni_alpha}, respectively. The algorithm then evaluates $\mathsf{P_{e}^{\eta_{2}}}$, i.e., $\mathsf{P_{e}^{\star}}$ at $\left\{\eta_{1}, \eta_{2}^{i}, \alpha\right\}$ and $\mathsf{P_{e}^{\alpha}}$, i.e., $\mathsf{P_{e}^{\star}}$ at $\left\{\eta_{1}, \eta_{2}, \alpha^{i}\right\}$. If for a given $\eta_{1}$,  $\left\vert\mathsf{P_{e}^{\alpha}}-\mathsf{P_{e}^{\eta_{2}}}\right\vert < \delta_{\mathsf{P_{e}}}$, then the algorithm exits the inner while-loop with $\mathsf{P_{e}^{\iota}}$ such that $\mathsf{P_{e}^{\iota}} = \min\left(\mathsf{P_{e}^{\alpha}}, \mathsf{P_{e}^{\eta_{2}}}\right)$ else, the algorithm iteratively descents in the steepest direction with new values of $\eta_{2}$ and $\alpha$. After traversing several values of $\eta_{1}$, the algorithm finally stops when for a given $\eta_{1}$, the obtained $\mathsf{P_{e}^{\iota}}$ is within $\delta_{\mathsf{P_{e}}}$ resolution of the previously computed value. The points at which $\mathsf{P_{e}^{\star}}$ is minimum as per the algorithm are given by $\eta_{1}^{\dagger}$, $\eta_{2}^{\dagger}$ and $\alpha^{\dagger}$. We rearrange the constraint in~\eqref{eq:constraint} to obtain $\epsilon_{2}^{\dagger}=2-\alpha^{\dagger}\left(\eta_{1}^{\dagger} + \eta_{2}^{\dagger}-2\right)$. Further, in the beginning we fixed $\epsilon_{1}=0$, therefore, $\epsilon_{1}^{\dagger}=0$. Thus, the algorithm computes all the 5 variables, i.e., $\epsilon_{1}^{\dagger}$, $\epsilon_{2}^{\dagger}$, $\eta_{1}^{\dagger}$, $\eta_{2}^{\dagger}$, and $\alpha^{\dagger}$.

In the next section, we showcase the efficacy of our algorithm in finding the near-optimal constellation points at Charlie for NC-F2FD. We also discuss the complexity analysis of our proposed algorithm. 

\section{Simulation Result and Complexity Analysis}

\begin{figure}
		\centering
		\includegraphics[width=\columnwidth]{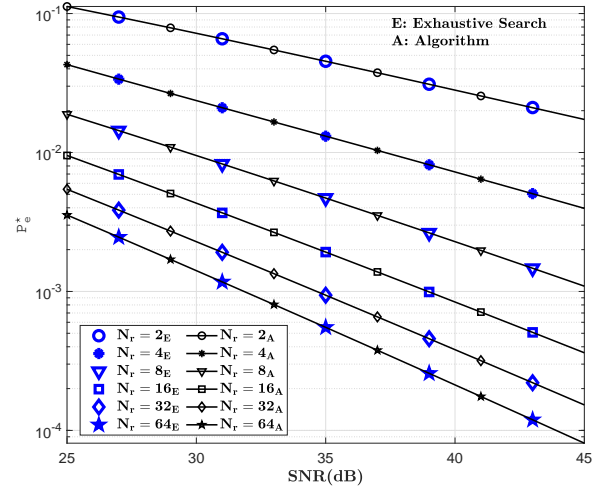}
		\caption{\label{fig:joint_error}Comparative error performance as a function of SNR, when constellation symbols are obtained using the Two-Layer Greedy Descent Algorithm ('A') and the exhaustive search ('E'). The subscript associated with the number of antennas denote the method used to obtain the constellation symbols, e.g., $\mathbf{N_{r}=2_{E}}$ means, for two receive antennas, exhaustive search is used to obtain the constellation symbols.}
	\end{figure}
	
\begin{figure}
		\centering
		\includegraphics[width=\columnwidth]{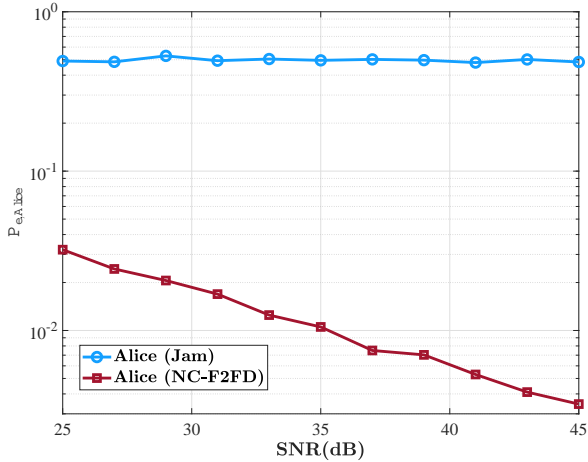}
		\caption{\label{fig:individual_error}Error performance of Alice for $N_{r}=4$, and $\Lambda=10^{-5}$ with and without NC-F2FD.}
		\end{figure}

In this section, we will demonstrate the error performance of NC-F2FD relaying scheme for various values of SNR and $N_{r}$. For all simulation purposes, we assume a practical FD radio at Charlie with SIC factor $\Lambda = 10^{-5}$~\cite{my_TCCN}. We first compute the optimal amplitude levels $\sqrt{\epsilon_{1}}$, $\sqrt{\epsilon_{2}}$, $\sqrt{\alpha\eta_{1}}$, and $\sqrt{\alpha\eta_{2}}$ that minimises $\mathsf{P_{e}^{\star}}$ using exhaustive search (denoted by $(\cdot)_{E}$ in~\figurename{\ref{fig:joint_error}}). We then plug these amplitude levels in~\eqref{eq:pe1} to obtain the error performance of NC-F2FD. We also compute the amplitude levels using the algorithm in the previous section. The constellation obtained using the algorithm is $\left\{\sqrt{\epsilon_{1}^{\dagger}}, \sqrt{\epsilon_{2}^{\dagger}}, \sqrt{\alpha^{\dagger}\eta_{1}^{\dagger}}, \sqrt{\alpha^{\dagger}\eta_{2}^{\dagger}}\right\}$. We again substitute these points in~\eqref{eq:pe1}, and plot the error performance for various SNR (denoted by $(\cdot)_{A}$ in~\figurename{\ref{fig:joint_error}}). It can be seen in~\figurename{\ref{fig:joint_error}} that the error performance of NC-F2FD exhaustive search and the algorithm overlaps, thus proving that the local minima $\left\{\epsilon_{1}^{\dagger}, \epsilon_{2}^{\dagger}, \eta_{1}^{\dagger}, \eta_{2}^{\dagger}, \alpha^{\dagger}\right\}$ is close to the global minima. Apart from plotting the error performance for various SNR, we also plot the error performance of NC-F2FD for various receive antennas at Bob. It is clear from the plots that as the receive-diversity increases, the error performance improves. Furthermore, in~\figurename{\ref{fig:individual_error}}, we plot the error performance of Alice when she uses power splitting factor $\alpha^{\dagger}$ as a function of SNR for 4 receive antennas at Bob, to showcase the efficacy of our scheme. It is evident that the error performance of Alice improves drastically after applying the countermeasure. However, in the process of aiding Alice, Charlie's performance degrades w.r.t. using a non-coherent signalling for point-to-point channel without using countermeasure~\cite{Ranjan}. 

In terms of complexity, the time taken by exhaustive search is very large compared to our algorithm. This is because the exhaustive search checks a number of quintets $(\alpha, \eta_{1}, \eta_{2}, \epsilon_{1}, \epsilon_{2})$  before computing the minima of $\mathsf{P_{e}^{\star}}$. The number of quintets depends on the resolution of each element. For instance, for the simulation results, $\delta_{\alpha} = 10^{-3}$, $\delta_{\eta_{1}} = \delta_{\eta_{2}} = 10^{-5}$, and $\delta_{\epsilon_{1}} = \delta_{\epsilon_{2}} = 10^{-4}$ give around 10 Million quintets for a given SNR and $N_{r}$. Here, $\delta_{(\cdot)}$ represents the step-size for increment of the element in subscript. In contrast, the Two-Layer Greedy Descent algorithm reduces the search space by using the results of Theorem~\ref{th:uni_eta2} and Theorem~\ref{th:uni_alpha}. Alternatively, we can also take a different approach and optimize $\mathsf{P_{e}^{\star}}$ w.r.t. $\eta_{1}$ and $\eta_{2}$ for a given $\alpha$. However, the time taken for the latter approach is greater than the former one, because with the former approach, the solution for $\eta_{1}$ is very close to the initialised value, i.e., $\eta_{1} = 0$. However, when we initialise the outer loop of the algorithm with $\alpha=0$ or $\alpha=1$, the algorithm takes more iteration as the optimal value of $\alpha$ is far from both $\alpha=0$ and $\alpha=1$. For instance, for a fixed SNR and $N_{r}$, the time taken by algorithms to converge to a solution when outer loop of algorithm is initialised $\eta_{1}=0$ and $\alpha=0$ is 1.2 seconds and 15 seconds, respectively. 

\section{Conclusion}

In this paper, we have presented DoS attack by a powerful FD \emph{jam and measure} adversary on nodes with low-latency constraint messages in fast-fading channels. We synthesised a fast-forward mitigation relaying scheme referred to as, NC-F2FD, where the victim and the helper node use modified non-coherent constellations to mitigate the jamming attack. Through rigorous analysis and simulation results, we showed that when both the nodes use NC-F2FD relaying scheme, the victim can not only successfully evade the jammer, but also compel the adversary to measure the same average power as before NC-F2FD.

\bibliographystyle{IEEEtran}
\bibliography{IEEEabrv, Ref}

\end{document}